\documentclass[runningheads]{llncs}
\usepackage{amssymb,bm,amsfonts,multicol,colortbl}
\usepackage{graphicx}
\usepackage{algorithm,algorithmic}
\usepackage{amsmath,amssymb,mathrsfs}
\usepackage{xcolor}
\usepackage{extarrows,mathrsfs}
\usepackage{tikz}
\usepackage[colorlinks, citecolor=blue]{hyperref}
 \usepackage[misc]{ifsym}
\usepackage{cite}
\newcolumntype{C}[1]{>{\centering}p{#1}}

\newcommand{\val}{\textnormal{val}}
\newcommand{\opt}{\textnormal{opt}}

\newcommand{\s}{\textsf}

\newcommand{\G}{\mathcal{G}}

\newcommand{\so}{\mathscr{D}_{1}}
\newcommand{\mbeta}{\bm{\beta}}
\newcommand{\mB} {\bm{B}}
\newcommand{\mS} {\bm{s}}
\newcommand{\bmx} {\bm{x}}
\newcommand{\myinit}{\mathrm{\Pi}_{\textnormal{PoC}}.\s{init}^H}
\newcommand{\myopen}{\mathrm{\Pi}_{\textnormal{PoC}}.\s{open}^H}
\newcommand{\myvrf}{\mathrm{\Pi}_{\textnormal{PoC}}.\s{vrf}^H}

\title{Towards a Multi-Chain Future of Proof-of-Space}
\author{Shuyang Tang\inst{1}\and Jilai Zheng\inst{1}\and Yao Deng\inst{1} \and Ziyu Wang\inst{2} \and Zhiqiang Liu\inst{1}\thanks{Zhiqiang Liu (\email{liu-zq@cs.sjtu.edu.cn}) and Dawu Gu (\email{gu-dw@cs.sjtu.edu.cn}) are the corresponding authors. The research is supported by the National Natural Science Foundation of China (Grant No. 61672347).}$^\text{(\Letter)}$ \and \\ Dawu Gu\inst{1}$^\text{(\Letter)}$}
\institute{$^1$Department of Computer Science and Engineering, Shanghai Jiao Tong University, Shanghai, China\\
\texttt{\{htftsy, zhengjilai, deng19930115\}@sjtu.edu.cn} \\
\texttt{\{liu-zq,gu-dw\} @cs.sjtu.edu.cn} \\ $^2$ School of Cyber Science and Technology, Beihang University, Beijing, China\\
\texttt{wangziyu@buaa.edu.cn}}

\authorrunning{Shuyang Tang et. al.}

\begin{document}
\maketitle

\begin{abstract}
Proof-of-Space provides an intriguing alternative for consensus protocol of permissionless blockchains due to its recyclable nature and the potential to support multiple chains simultaneously. However, a direct shared proof of the same storage, which was adopted in the existing multi-chain schemes based on Proof-of-Space, could give rise to newborn attack on new chain launching. To fix this gap, we propose an innovative framework of single-chain Proof-of-Space and further present a novel multi-chain scheme which can resist newborn attack effectively by elaborately combining shared proof and chain-specific proof of storage. Moreover, we analyze the security of the multi-chain scheme and prove that it is incentive-compatible. This means that participants in such multi-chain system can achieve their greatest utility with our proposed strategy of storage resource partition.
\keywords{Blockchain \and Mechanism Design \and Consensus \and Cryptocurrency \and Proof-of-Space}
\end{abstract}

\section{Introduction}
Since the proposal of Bitcoin in 2008~\cite{bitcoin},
blockchain has been successfully providing a decentralized consensus on a distributive ledger in a permissionless environment through a peer-to-peer network.
In a high level, blockchain is a chain of blocks, each containing certain linearly ordered transactions.
The consensus of blockchain performs a ``leader election'' process and a ``ledger extension'' process for each round.
The ``leader election'' process elects  one  or few  \emph{leaders} from all consensus participants (known as \emph{miners}) according to their computing power, and then these leaders perform a ``ledger extension'' process via appending their proposed blocks to the rear of the blockchain.
Since the specific computing power is evaluated by having miners work on finding certain hash functions preimages, such a way of leader election is called \emph{Proof-of-Work}~(PoW)~\cite{DBLP:conf/crypto/DworkN92}.
The computing power of conducting PoW is referred to as \emph{hash power}.

PoW consensus scheme suffers from notorious consumption of hash power which turns out to be a waste of natural resource. This circumstances has activated the investigation of alternative consensus schemes which are more energy-efficient. An alternative consensus scheme is \emph{Proof-of-Stake}~(PoS)~\cite{pos,DBLP:journals/sigmetrics/BentovLMR14,DBLP:journals/iacr/BentovPS16a,DBLP:conf/asiacrypt/PassS17,DBLP:conf/crypto/KiayiasRDO17,DBLP:conf/eurocrypt/DavidGKR18,DBLP:conf/ccs/BadertscherGKRZ18}. In the ``leader election'' process of PoS consensus scheme, the leader is randomly selected proportionally to the \emph{stake} that each miner holds, rather than the hash power.
Another alternative to PoW is \emph{Proof-of-Space} (aka. \emph{Proof-of-Capacity}, PoC).
In PoC consensus, the one-time consumption of hash power is replaced by the holding of the storage resource, namely, recyclable hardware disks.
Also, PoC is inherited with a nature of concurrency by
providing resource proof for multiple chains simultaneously.
With a shared nonce for the \emph{pebbling graph}~\cite{DBLP:conf/crypto/DworkNW05,DBLP:conf/crypto/DziembowskiKW11,DBLP:conf/tcc/DziembowskiKW11}~(PG, the building block of most existing PoC-based consensus schemes),
the claim on the same storage can be applied to the consensus of more than one PG-based blockchain.
With such a shared proof of storage, the same storage contributes to the security of
multiple chains -- less total resource is required for the same security guarantee globally.

However, a direct shared proof of the same storage brings about \emph{newborn attack},
which makes new chains hard to launch since holders of large storage pools may attack a newly started chain with almost zero-cost~(it needs only to duplicate the proof one more time).
In this paper, we aim to address this issue by proposing an innovative multi-chain scheme of proof-of-space based on
the SpaceMint~\cite{spacemint} protocol. This scheme is built on a combination of a shared proof and a chain-specific proof of storage, which makes the same storage source contributes simultaneously to multiple blockchains, and the cost for an adversary to launch a \emph{newborn attack} is enormous. Moreover, we prove that our scheme is incentive-compatible,  which means
participants can achieve their greatest utility
with our desired strategy of storage resource partition.
\subsection{Related Works}
 Two recent works~\cite{8,2}both proposed their ``Proof-of-Space" protocol. Their main difference is whether the proofs are for transient storage or persistent storage. Proof of transient space(PoTS) by Ateniese et al.~\cite{8} enables efficient verification of memory-hard function~\cite{7}, which is a function that needs a lot of space to compute. In this work, a verifier only needs $O(\s{poly} \log S)$ time and space to verify the claimed space usage of the prover, where $S$ is the amount of storage the prover $P$ wants to dedicate.  Proof of persistent space(PoPS) by Dziembowski et.  al.~\cite{2} allows the verifier to repeatedly audit the prover.  Only prover who stores the data persistently can pass the repeatedly audit.
\par
Proof of retrievability(PoR)~\cite{9} allows a user who outsources some useful data to an untrusted server to repeatedly verify if the data is still existing in the server. The difference between PoPS and PoR is whether a large amount of initial data is transferred from verifier $\textsf{V}$ to prover $\textsf{P}$.
\par
Previous PoC constructions only differ in the pebbling graphs $G$ .
Dziembowski et.  al.~\cite{2} proposed two constructions of PoC schemes in the random oracle model, using Merkle tree and graphs with high ``pebbling complexity". One is based on a graph with high pebbling complexity by Paul et. al.~\cite{3}, which is
$(\Theta(S/\log S), S/\log S, \infty)$-secure~(see definition in Sec.~\ref{sec:gua}).
Another one combines superconcentrators~\cite{DBLP:journals/im/AlonC03,DBLP:conf/sirocco/Schoning97,DBLP:journals/arscom/KolmogorovR18},  random bipartite expander graphs~\cite{DBLP:journals/ppl/Haglin95,DBLP:journals/dm/Thomason89a}
and depth robust graphs~\cite{4,DBLP:conf/eurocrypt/AlwenBP17} and is
\begin{center}
$(\Theta(S), \infty, \Theta(S))$-secure
\end{center}
\par
The construction of Ren and Devadas~\cite{5} uses stacked bipartite expander graphs, which is
\begin{center}
$(\alpha \cdot S, (1-\alpha) \cdot S, \infty)$-secure
\end{center}
for any $\alpha \in [0, 0.5]$. The labels of the graph are computed just as in ~\cite{2}, however the prover only stores the labels on top of this stack.
\par
The construction of Krzysztof Pietrzak~\cite{6} uses depth-robust graphs from ~\cite{7} to realize PoC,  which is

\begin{center}
$(S \cdot (1-\epsilon), \infty, S)$-secure
\end{center}
This construction has a tight bound, which means it can get security against adversary storing $(1-\epsilon)$ fraction of the space.  Moreover,  this construction gets security against parallelism, which implies massive parallelism of oracle queries doesn't benefit the adversary. Besides, this work also introduces and constructs a new type of PoC, which allows the prover to store useful data at the same time.

\subsection{Paper Organization}
The remainder of this paper is organized as follows.
Sec.~\ref{sec:2} introduces notations, building blocks and the background of PoC.
Sec.~\ref{sec:4} presents a single-chain PoC scheme built on SpaceMint -- one of the most well-known PoC scheme today.
Based on this, in Sec.~\ref{sec:5}, we describe our framework for the multi-chain PoC scheme both by a general functionality and a specific protocol realizing one case of the functionality.
Finally, we analyze the incentive compatibility and security of our framework in Sec.~\ref{sec:6}.

\section{Backgrounds} \label{sec:2}

\subsection{Notations}
For a set $S$,  $|S|$ denotes the number of elements in $S$. With ``$||$'' we denote concatenation of strings.
More generally, for any two tuples $\bm{m}_1,\bm{m}_2$, $\bm{m}_1\circ \bm{m}_2$ is the concatenation of them and for unary tuple $\bm{m}_1=(m)$, $\bm{m}_1\circ \bm{m}_2$ is written as $m\circ \bm{m}_2$ (same to the case of $\bm{m}_2=(m)$).
We ideally assume $N$ participants and refer to them by either identities $(P_1,P_2,\ldots,P_N)$ or their public keys ($\s{pk}_1, \s{pk}_2,\ldots, \s{pk}_N$) interchangeably.
The secret key corresponding to a public key $\s{pk}$ is denoted as $\s{pk}^{-1}$ for simplicity.
To further facilitate our description of a high-level framework.
We assume a public-key infrastructure (PKI) among all participants,
which is described by a functionality $\mathcal{F}_{cert}$.
Moreover, to any tuple of messages $(m_0,m_1,\ldots,m_\ell)$, we use $(m_0,m_1,\ldots,m_\ell)_{\s{pk}^{-1}}$ to signify the tuple along with the valid signature on the tuple hash from the participant of public key $\s{pk}$.
In later descriptions, we may take the necessity of signing and verifying as granted and avoid redundant descriptions.
We assume a hash function $H: \{0,1\}^*\to\{0,1\}^\lambda$ simulating the random oracle~\cite{DBLP:conf/ccs/BellareR93}.
Note that based on $H$, we can build $H_z$ for any message $z$ as $H_z(\cdot):=H(z||\cdot)$~\cite{10}.
%
%\subsection{Merkle-tree Commitment}
%A commitment scheme contains three algorithms: $commit, open, verify$.
% $H:\{0, 1\}^* \rightarrow \{0, 1\}^d$ is a collision-resistant hash function. When a party $P$ wants to commit to values $\boldsymbol{y}=(y_1, . . , y_n)$, it invokes
% $$
%(\gamma_{\boldsymbol{y}}, \gamma):=commit^H(\boldsymbol{y})
%$$
%where  $\gamma_{\boldsymbol{y}} \in \{0, 1\}^d$ is the value of the Merkle root, which is the commitment.  $\gamma$ denotes the $2n-1$ values of all nodes in the tree, which will be stored by $P$. We define $y_{ab}=H(y_a||y_b)$
%\par
%Once $P$ announces $\gamma_{\boldsymbol{y}}$ it is committed to $\boldsymbol{y}$. Then it can open any commmitted value $y_i$ by invoking
%$$
%\boldsymbol{o}:=open^H(\gamma, i) \in \{0, 1\}^{\lceil \log (n) \rceil \cdot d}
%$$
%The $\log n$ values relate to the siblings of the nodes which are on the path from $y_i$ to the root.
%\par
%Anyone who invokes $verify^H(\gamma_{\boldsymbol{y}}, y_i, \boldsymbol{o})$ can check if $\boldsymbol{o}$ is the right opening.
%\par
%For instance,  for $n=4$, the commitment to $y_1, . . , y_4$ would be:
%$$
%\gamma_{\boldsymbol{y}}=H(H(y_1||y_2)||H(y_3||y_4))
%$$
%The corresponding opening to $y_2$ should be :
%$$
%\boldsymbol{o}=open^H(\gamma, 2)=(y_1, y_{34})
%$$
%Anyone can check if
%$$
%verify^H(\gamma_{\boldsymbol{y}}, y_2, (y_1, y_{34}))=\left( H( H(y_1||y_2)||y_{34}) \overset{?}{=} \gamma_{\boldsymbol{y}} \right)
%$$

\subsection{Proof-of-Space}\label{sec:gua}
Proof-of-space is an interactive protocol between a prover $\textsf{P}$ and a verifier $\textsf{V}$ that demonstrates
the prover $\textsf{P}$ is storing some data of a certain size.
The PoC protocol in~\cite{2} involves two phases: initialization phase and  execution phase.
\par
Initialization is an interactive process between the prover $\textsf{P}$ and the verifier $\textsf{V}$. It runs on shared inputs $(id, S)$.  $id$ is an identifier to assure that the prover $\textsf{P}$ cannot reuse the same disk space to run PoC for different statement.  $S$ is the amount of storage the prover $\textsf{P}$ wants to dedicate.
After the initialization phase,  \s{P} stores some data \s{F} ,  whereas \s{V} only stores a commitment $\gamma$ to \s{F}.
\par
Execution is an interactive process between the prover $\textsf{P}$ and the verifier $\textsf{V}$. The prover $\textsf{P}$ runs on data \s{F} and the verifier $\textsf{V}$ runs on input $\gamma$. Then the verifier $\textsf{V}$ sends challenges to the prover $\textsf{P}$,  obtains back  the corresponding openings. At the end $\s{V}$ verifies these openings and outputs $accept$ or $reject$.

The security of a PoC protocol was formally defined in~\cite{2}. Specifically, a $(S_0, S_1, T)$-adversarial prover $\hat{\textsf{P}}$ was defined, which means $\hat{\textsf{P}}$'s storage after the initialization phase is bounded by $S_0$, while during the execution phase it runs in storage at most $S_1$ and time at most $T$.  Secure PoC protocols are required to have three properties, which are completeness,  soundness and efficiency.
\begin{enumerate}
\item Completeness: We say a PoC protocol has completeness if the verifier always outputs $accept$ for any honest prover $\textsf{P}$ with probability $1$.
\par
\item Soundness: We say a PoC protocol has soundness if the verifier $\textsf{V}$ outputs $accept$ with a negligible probability for any $(S_0, S_1, T)$-adversarial prover $\hat{\textsf{P}}$ .
\par
\item Efficiency: We say a PoC protocol has efficiency if the verifier $\textsf{V}$ can run in time $O(\s{poly} \log S)$.
\end{enumerate}
We say that a PoC protocol is $(S_0, S_1, T)$-secure if the above three properties are satisfied.

\section{Single-Chain Proof-of-Space}  \label{sec:4}
In this section, we use SpaceMint protocol as the building block
of our multi-chain protocol to be described in the latter section.

\subsection{Graph Labeling Game}
At first, we introduce the graph labeling game~\cite{posw,DBLP:books/daglib/0016239}.
\begin{definition}[Graph Labelling]
We consider a directed acyclic graph (DAG)
$G = (V,E)$ ,which has a vertex set $V = \{0,1,\ldots,S-1\}$ and a hash function $H : \{0,1\}^* \rightarrow \{0,1\}^\lambda$, the label $l_i \in \{0,1\}^\lambda$ for each vertex $i \in V$is recursively computed as $l_i = H(\s{nc},i,l_{p_1},...,l_{p_t})$
where $l_{p_1},...,l_{p_t}$ are the parents of vertex $i$ and $\s{nc}$ is a unique nonce.
\end{definition}
 Graph labeling game is the building block of most PoC protocols. Let $G = (V, E)$ be a directed acyclic graph (DAG) which has $S$ nodes . $H : \{0, 1\}^* \rightarrow \{0, 1\}^\lambda$ is a collision-resistant hash function. For every $id$, a fresh hash function can be sampled: $H_{id} = H(id||\cdot)$.  The PoC protocol based on the graph labeling game is as follows:
\par
\begin{enumerate}
\item Initialization:  First, $\textsf{P}$ computes the labels on all nodes
of $G$ using graph labelling,  and commits to them in $\gamma$ using Merkle-tree Commitment.  Then $\textsf{P}$
gets $q$ challenges from $\textsf{V}$.  For each of the challenges $C =(C_1, . . . , C_q)$,  $\textsf{P}$ opens the label on the $C_i^{\text{th}}$ node
of $G$, as same as the labels of all its parent nodes. Finally, $\textsf{V}$ checks the openings of all challenge nodes and their parents. $\textsf{V}$ also checks if the challenge nodes are computed correctly from their parents.
\par
\item Execution: $\textsf{V}$ chooses $t$ challenge nodes randomly.
Then $\textsf{P}$ sends the opening of these challenge nodes to $\textsf{V}$. This execution phase can be done repeatedly.
\end{enumerate}
%
%Proof-of-Space~(PoC) is a protocol between a prover \s{P} and a verifier \s{V} . It has two phases: initialization phase and  execution phase. After the initialization phase, \s{P} stores some data \s{F}, whereas \s{V} only stores a commitment
%$\gamma$ to \s{F}. In the execution phase, \s{V} sends some challenges \s{C} to \s{P}, then \s{P} returns the corresponding openings to \s{V}, and at the end \s{V} verifies these openings.\\

%During the initialization phase, the prover P chooses a directed acyclic graph $G=(V,E)$ depending on the amount of storage he wants to dedicate.Then the prover P computes and stores all the labels in the graph by using the hash function.

\subsection{PoC Definition}
Before introducing how to apply any PoC schemes to blockchains, we introduce the formal description of a PoC
at first designed to apply between two parties.
A PoC scheme built on a hash oracle $H$ is described as $\mathrm{\Pi}_{\text{PoC}}=(\s{init},\s{open},\s{vrf})$.
Specifically,
\begin{description}
  \item[Space Commitment.] $\myinit(S)\to (\gamma,\widetilde{\gamma})$ inputs the size of the pebbling graph and returns a pair $(\gamma,\widetilde{\gamma})$ after building up the graph where $\gamma$ is the commitment released to the publicity and $\widetilde{\gamma}$ is a secret to be locally stored.
  \item[Commitment Opening.] $\myopen(S,\widetilde{\gamma},C)\to\tau$ takes as input the graph size, the threshold $\widetilde{\gamma}$, the challenge $C$ and returns a proof $\tau$.
  \item[Verification.] $\myvrf(S,\gamma,C,\tau)\to\{accept, reject\}$ verifies a proof $\tau$ for the challenge $C$ to the graph of size $S$ and public commitment $\gamma$.
\end{description}
%\begin{definition}[PoC]
%{\bfseries 1.Common Inputs}: \\
%  $H$: a random oracle $\{0,1\}^* \rightarrow \{0,1\}^\lambda$\\
%  $d$ :statistical security parameter \\
%  $t$ :statistical security parameter,the number of challenges \\
%  $S$ :a space parameter \\
%{\bfseries 2.Compute PoC}:\\
% \s{P} computes a proof $(\gamma, \gamma_P) := PoC^H(S)$. \s{P} sends $\gamma$ to \s{V} and locally
%stores $\gamma_P$.\\
%{\bfseries 3.Challenge}:\\
% \s{V} samples a random challenge $C =(C_1,...,C_t)$ and sends it to \s{P}.\\
%{\bfseries 4.Open}:\\
% \s{P} computes $\tau =(\tau_1,...,\tau_t):= open^H(S,\gamma_P,C)$ and sends it to \s{V}.\\
%{\bfseries 5.Verify}:\\
% \s{V} computes and outputs $verify^H(N,\gamma,C,\tau) \in \{accept,reject\}.$
%\end{definition}

\subsection{SpaceMint}
In SpaceMint~\cite{spacemint}, the structure of blocks is identical to Bitcoin blockchain except for containing a space proof instead of a hash solution.
Furthermore, the detailed protocol of SpaceMint is similar to any blockchain, in spite that the mining and chain-competition differ from that of PoW-based blockchains.
Thereby,
it is too redundant to cover every details to introduce the full scheme.
To describe SpaceMint, we only need to enumerate all the difference between SpaceMint and the ordinary blockchain.
\begin{description}
  \item[Initial Step.] For a miner to dedicate $\lambda N$ bits of storage to the blockchain network,
  it computes and stores the labels of their pebbling graph
  to get $(\gamma,\widetilde{\gamma})\gets\myinit(S)$ at the initial stage. Afterwards $sctx=(\s{pk},\gamma)_{\s{pk}^{-1}}$ is broadcast to the network.

  \item[The Mining.]
  Each miner maintains a main chain (the chain branch with the greatest total weight) in its view.
  We denote the chain as $(A_1, A_2, \ldots, A_{i-1})$ and their corresponding proofs as $\tau^1,\tau^2,...,\tau^{i-1}$~($i\ge\Delta$).
  To mine the next block $A_i$, the miner at first  derives the challenge for block $i$ from
the proof of block $i-\Delta$, i.e., $C_i =H(\tau_{{i-\Delta}})  \mod  S$.
    Thereby, the miner obtains $\tau\gets\myopen(S,\widetilde{\gamma},C_i)$ and assembles $A_i$ with $\tau$. This differs from PoW.

\item[Block Verification.]
    Different from the nonce verification in the existing blockchain, we in SpaceMint
    should check the open-up $\tau$ of each block by $\myvrf(S,\gamma,C,\tau)$.

  \item[Chain Weight.]
  Assuming a block is followed by $m$ sequential blocks $(A_1, A_2, \ldots, A_m)$, which have valid proofs $\tau^1,\tau^2,...,\tau^m$ respectively, and that the corresponding space contributed by these miners are $S^1,S^2,...,S^m$. The weight of $A_i$ is defined as follows:
$$
\s{weight}(A_i)=(\frac{H(\tau^i)}{2^\lambda})^{1/{S^i}}
$$
It can be proved that the probability that
$A_i$ has the largest weight among these $m$ blocks equals to his fraction of the total space, which is
$\frac{S^i}{\sum^m_{k=1} S^k}$.
In this way, in a chain competition, the chain branch with the greatest total weight outruns the others.
This differs from the existing blockchain where the block weight is atomic (either one or zero).
\end{description}

\section{A Multi-Chain scheme}  \label{sec:5}
In this section, we introduce our framework for the multi-chain scheme of proof-of-space.
To this end, we first describe our high-level functionality and then propose one possible realization of the functionality based on the
aforementioned SpaceMint-based single-chain protocol.

\subsection{Functionalities}

\begin{figure}[t]
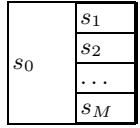

  \centering
  \begin{tabular}{|p{0.74cm} p{0.74cm}|}
    \hline
    $s_0$ &
    \begin{tabular} {|p{0.7cm}|}
    \hline
    $s_1$ \\
    \hline
    $s_2$ \\
    \hline
    $\ldots$ \\
    \hline
    $s_M$ \\
    \hline
    \end{tabular} \\
    \hline
  \end{tabular}
  \caption{Capacity Resource Partition}\label{fig:tabu}
\end{figure}

\begin{figure*}[h]
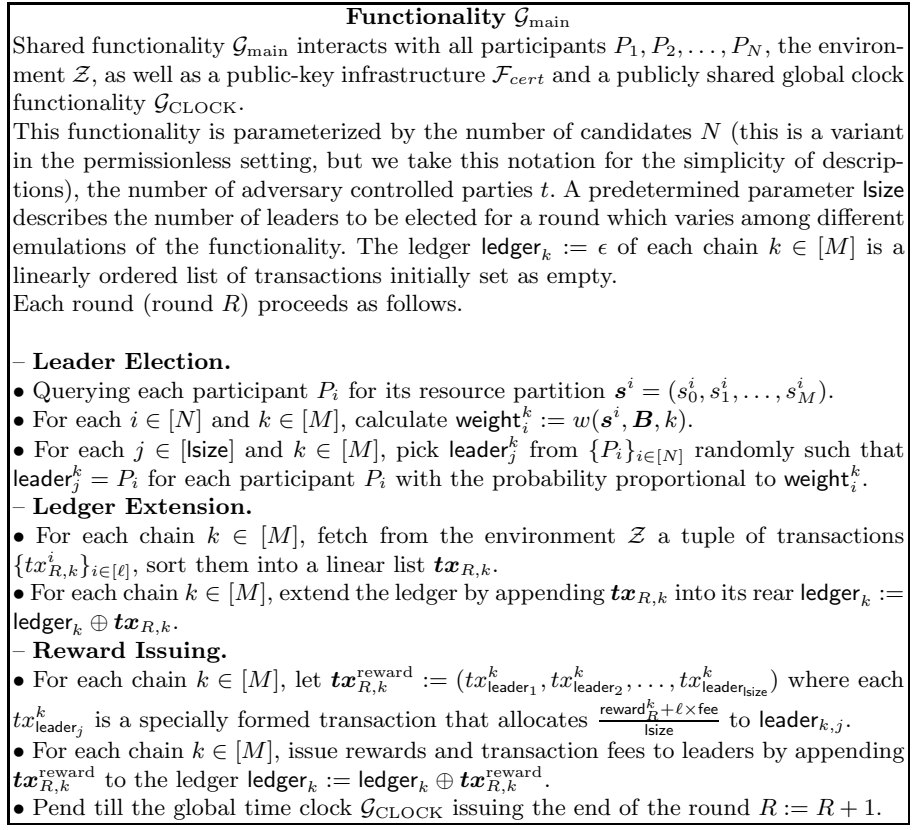

 \centering

\resizebox{\textwidth}{!}{ %
\begin{tabular}{|p{12cm}|} \hline
\multicolumn{1}{|c|}{\textbf{Functionality} $\G_{\text{main}}$}\\
Shared functionality $\G_{\text{main}}$ interacts with all participants $P_1,P_2,\ldots,P_N$, the environment $\mathcal{Z}$, as well as a public-key infrastructure $\mathcal{F}_{cert}$ and a publicly shared global clock functionality $\G_{\text{CLOCK}}$.

This functionality is parameterized by the number of candidates $N$ (this is a variant in the permissionless setting, but we take this notation for the simplicity of descriptions), the number of adversary controlled parties $t$. A predetermined parameter $\s{lsize}$ describes the number of leaders to be elected for a round which varies among different emulations of the functionality. The ledger $\s{ledger}_k:=\epsilon$ of each chain $k\in[M]$ is a linearly ordered list of transactions initially set as empty.

Each round (round $R$) proceeds as follows.\\
\\
-- \textbf{Leader Election.}\\
$\bullet$ Querying each participant $P_i$ for its resource partition $\bm{s}^i=(s_0^i, s_1^i, \ldots, s_M^i)$.\\

$\bullet$ For each $i\in[N]$ and $k\in[M]$, calculate $\s{weight}_i^k:= w(\bm{s}^i,\bm{B},k)$.\\

$\bullet$ For each $j\in[\s{lsize}]$ and $k\in[M]$, pick $\s{leader}^k_j$ from $\{P_i\}_{i\in[N]}$ randomly such that $\s{leader}^k_j=P_i$ for each participant $P_i$  with the probability proportional to $\s{weight}_i^k$.
\\
-- \textbf{Ledger Extension.}\\
$\bullet$ For each chain $k\in[M]$, fetch from the environment $\mathcal{Z}$ a tuple of transactions $\{tx_{R,k}^i\}_{i\in[\ell]}$, sort them into a linear list $\bm{tx}_{R,k}$. \\
$\bullet$ For each chain $k\in[M]$, extend the ledger by appending $\bm{tx}_{R,k}$ into its rear $\s{ledger}_k:=\s{ledger}_k\oplus \bm{tx}_{R,k}$.
\\
-- \textbf{Reward Issuing.}\\
$\bullet$ For each chain $k\in[M]$, let $\bm{tx}_{R,k}^{\textnormal{reward}}:=(tx_{\s{leader}_1}^k,tx_{\s{leader}_2}^k,\ldots,tx_{\s{leader}_\s{lsize}}^k)$ where each $tx_{\s{leader}_j}^k$ is a specially formed transaction that allocates $\frac{\s{reward}^k_R+\ell\times\s{fee}}{\s{lsize}}$ to $\s{leader}_{k,j}$.
\\
$\bullet$ For each chain $k\in[M]$, issue rewards and transaction fees to leaders by appending $\bm{tx}_{R,k}^{\textnormal{reward}}$ to the ledger $\s{ledger}_k:=\s{ledger}_k\oplus \bm{tx}_{R,k}^{\textnormal{reward}}$. \\
$\bullet$ Pend till the global time clock $\G_{\text{CLOCK}}$ issuing the end of the round $R:= R+1$.\\
 \hline
\end{tabular}
 }%
\caption{The Main Functionality}
\label{fig:mpoc_functionality}
\end{figure*}

\begin{figure*}[b]
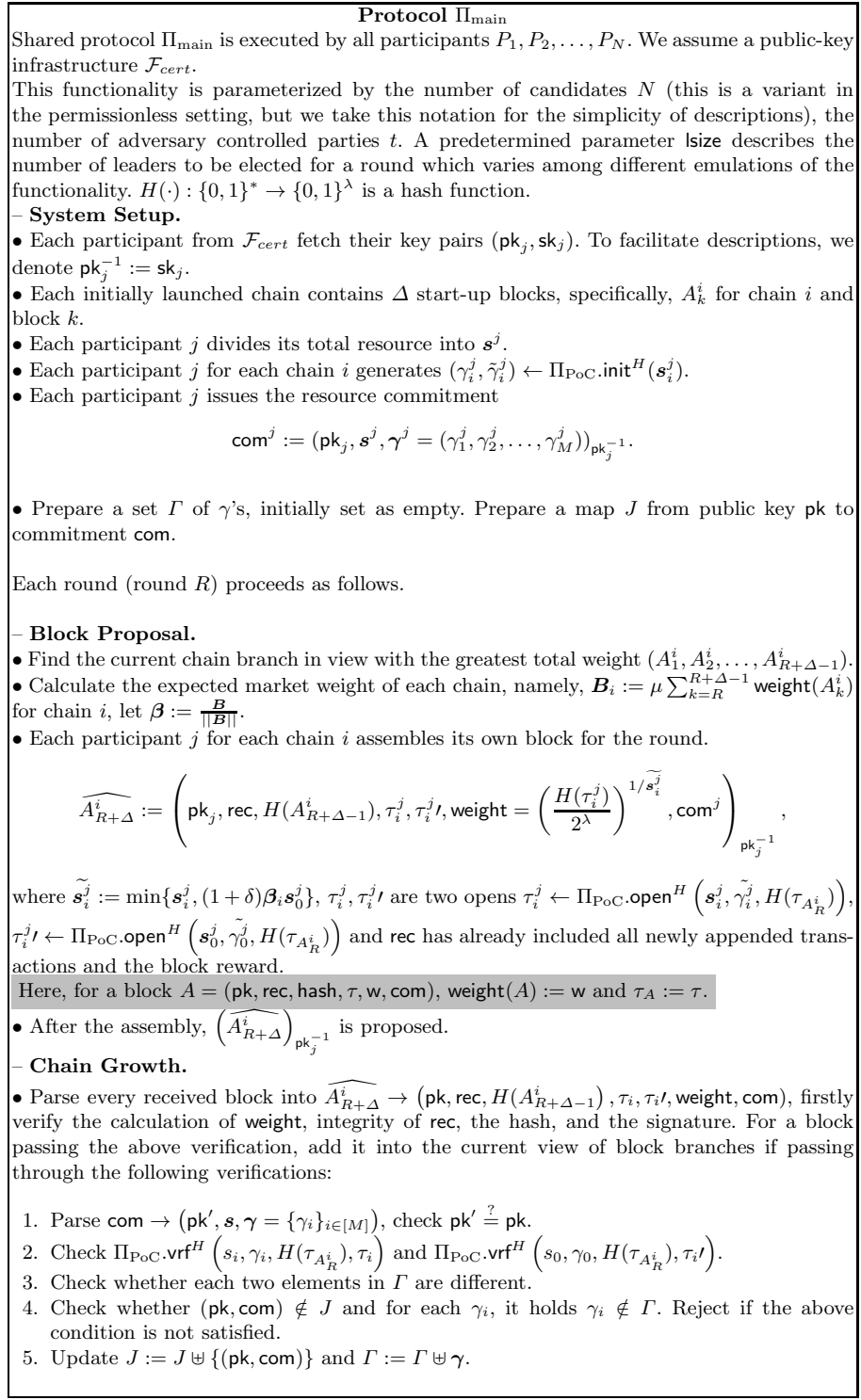

 \centering

\resizebox{\textwidth}{!}{ %
\begin{tabular}{|p{13cm}|} \hline
\multicolumn{1}{|c|}{\textbf{Protocol} $\mathrm{\Pi}_{\text{main}}$}\\
Shared protocol $\mathrm{\Pi}_{\text{main}}$ is executed by all participants $P_1,P_2,\ldots,P_N$. We assume a public-key infrastructure $\mathcal{F}_{cert}$.

This functionality is parameterized by the number of candidates $N$ (this is a variant in the permissionless setting, but we take this notation for the simplicity of descriptions), the number of adversary controlled parties $t$. A predetermined parameter $\s{lsize}$ describes the number of leaders to be elected for a round which varies among different emulations of the functionality.
$H(\cdot): \{0, 1\}^*\to \{0,1\}^\lambda$ is a hash function.\\

-- \textbf{System Setup.}\\
$\bullet$ Each participant from $\mathcal{F}_{cert}$ fetch their key pairs $(\s{pk}_j,\s{sk}_j)$. To facilitate descriptions, we denote $\s{pk}_j^{-1}:=\s{sk}_j$.\\
$\bullet$ Each initially launched chain contains $\Delta$ start-up blocks, specifically, $A_k^i$ for chain $i$ and block $k$.\\
$\bullet$ Each participant $j$ divides its total resource into $\mS^j$.\\
$\bullet$ Each participant $j$ for each chain $i$ generates $(\gamma^j_i,\tilde{\gamma}_i^j)\gets \myinit(\mS_i^j)$.\\
$\bullet$ Each participant $j$ issues the resource commitment
$$
\s{com}^j:=(\s{pk}_j,\mS^j,\bm{\gamma}^j=(\gamma_1^j,\gamma_2^j,\ldots,\gamma_M^j))_{\s{pk}_j^{-1}}.
$$\\
$\bullet$ Prepare a set $\Gamma$ of $\gamma$'s, initially set as empty. Prepare a map $J$ from public key $\s{pk}$ to commitment $\s{com}$.
%$\bullet$ Dynamically receive from other participants $j'$ resource commitments $\s{com}^{j'}$. Parse $\s{com}^{j'}\to (\s{pk}_{j'},\mS^{j'},\bm{\gamma}^{j'},\s{sig}^{j'}_{com})$ reject if $\mS^{j'}$ is not an admissible partition.

%$\bullet$ The global market weight of all chains is calculated $\mB:=\sum_{j=1}^N \mS^i$.

\\
Each round (round $R$) proceeds as follows.\\
\\
-- \textbf{Block Proposal.}\\
$\bullet$ Find the current chain branch in view with the greatest total weight $(A^i_{1},A^i_{2},\ldots,A^i_{R+\Delta-1})$.\\
$\bullet$ Calculate the expected market weight of each chain, namely, $\mB_i:=\mu\sum_{k=R}^{R+\Delta-1}{\s{weight}(A_k^i)}$ for chain $i$, let $\mbeta:=\frac{\mB}{||\mB||}$.\\
$\bullet$ Each participant $j$ for each chain $i$ assembles its own block for the round.
$$
\widehat{A_{R+\Delta}^i} := \left( \s{pk}_j, \s{rec}, H(A^i_{R+\Delta-1}), \tau_i^j, \tau_i^j\prime, \s{weight}=\left(\frac{H(\tau_i^j)}{2^\lambda}\right)^{{1}/\widetilde{\mS_i^j}} , \s{com}^j\right)_{\s{pk}_j^{-1}},
$$
where $\widetilde{\mS_i^j}:=\min\{\mS_i^j, (1+\delta)\mbeta_i\mS_0^j\}$, $\tau_i^j, \tau_i^j\prime$ are two opens $\tau_i^j\gets \myopen\left(\mS_i^j,\tilde{\gamma_i^j},H(\tau_{A^i_R})\right)$,
$\tau_i^j\prime\gets \myopen\left(\mS_0^j,\tilde{\gamma_0^j},H(\tau_{A^i_R})\right)$ and $\s{rec}$ has already included all newly appended transactions and the block reward.

\colorbox{lightgray}{Here, for a block $A=(\s{pk},\s{rec},\s{hash},\tau,\s{w},\s{com})$, $\s{weight}(A):=\s{w}$ and $\tau_A:=\tau$.}\\

$\bullet$ After the assembly, $\left(\widehat{A_{R+\Delta}^i}\right)_{\s{pk}_j^{-1}}$ is proposed.
\\
-- \textbf{Chain Growth.}\\
$\bullet$ Parse every received block into $\widehat{A_{R+\Delta}^i} \to \left(\s{pk},\s{rec},H(A^i_{R+\Delta-1}\right), \tau_i, \tau_i\prime, \s{weight},\s{com})$, firstly verify the calculation of  $\s{weight}$, integrity of $\s{rec}$, the hash, and the signature. For a block passing the above verification, add it into the current view of block branches if passing through the following verifications:
\begin{enumerate}
  \item Parse $\s{com}\to\left(\s{pk}',\mS,\bm{\gamma}=\{\gamma_i\}_{i\in[M]}\right)$,
      check $\s{pk}'\stackrel{?}{=}\s{pk}$.
  \item Check $\myvrf\left(s_i,\gamma_i,H(\tau_{A^i_R}),\tau_i\right)$ and $\myvrf\left(s_0,\gamma_0,H(\tau_{A^i_R}),\tau_i\prime\right)$.
  \item Check whether each two elements in $\Gamma$ are different.
  \item Check whether $(\s{pk},\s{com})\notin J$ and for each $\gamma_i$, it holds $\gamma_i\notin\Gamma$. Reject if the above condition is not satisfied.
  \item Update $J:=J\uplus\{(\s{pk},\s{com})\}$ and $\Gamma:=\Gamma\uplus\bm{\gamma}$.
\end{enumerate}
\\
 \hline
\end{tabular}
 }%
\caption{The Main Protocol}
\label{fig:mpoc_protocol}
\end{figure*}

In a high level, as Fig.~\ref{fig:tabu}, we expect to have all participants pay half of the total storage for purchasing the upper bound of their admissible storage (the shared proof part, denoted as $\mS_0$),
and allocate the rest half of storage on each supportive chains respectively according to their market weight~($\mS_i$ for chain $i$).
Specifically, to a chain $i$ with $\mbeta_i=\frac{\mB_i}{||\mB||}$ ($\mB_i$ is the sum of $\mS_i$ for all participants) fraction of total storage proof, we stimulate the participant to allocate $\mS_i=\mbeta_i\mS_0$ storage on it.
To allow for the start-up of new chains, we in actual ask that $\mS_i\le(1+\delta)\mbeta_i\mS_0$.
Clearly, as an equilibrium of economics, the total storage proof behind a blockchain is proportional to the total market weight of the chain.
Therefore, the market weight of each chain is described by the total storage proof behind the chain.
Also, we assume that the amount of released coins and transaction fees is the same to all chains for each round and hence the per-coin value is proportional to our defined market weight.

In all, our system supports $M$ different PoC-based blockchains.
Each ($j^{\textnormal{th}}$) of them elects leaders according to $s^i_j$.
Specifically, each miner ($i^{\textnormal{th}}$) has the chance of being the leader of each round by
$$
w(\bm{s},\bm{B},k) := \left\{ \begin{matrix}
\frac{\bm{s}_k}{\mB_k+\mS_k} , & \mS_k\le (1+\delta)\mbeta'_k\mS_0 \\
\frac{\widetilde{\bm{s}_k}}{\mB_k+\widetilde{\mS_k}} , & \textnormal{otherwise}\\
 \end{matrix} \right.
$$
where $\mB'_k := \mB_k+\mS_k$, $\mbeta':=\frac{\mB'}{||\mB'||}$,
and $\widetilde{\bm{s}_k} := \max_{\mS_k}\{\mS_k\le(1+\delta)\mbeta_k'\mS_0\}$.
The exact formula of $\mS_k$ is easy to uncover by solving a quadratic equation (hence can be described in codes when implemented). However, we remain the current formulation for readability.
To facilitate proofs in later sections, we observe that $\psi_{\mB,k}(\mS):=w(\mS,\mB,k)$ is continuous on $(\mathbb{R}^+)^{M+1}$ and is smooth almost everywhere (a.e.) on $(\mathbb{R}^+)^{M+1}$ except for $\{\mS\in(\mathbb{R}^+)^{M+1}|\mS_k=(1+\delta)\mbeta_k'\mS_0 \}$ with zero Lebesgue-measure for each $k\in[M]$.

The detailed description of the functionality is shown in Fig.~\ref{fig:mpoc_functionality},
where we aim to provide a framework more general than our specific way of realization so the size of leaders to be elected for each round is parameterized by $\s{lsize}$.
In fact, $\s{lsize}$ of most existing blockchains and our realization based on SpaceMint is simply $1$.

\subsection{A Protocol for $\s{lsize}=1$}
To build our realization of the above functionality with $\s{lsize}=1$ based on our SpaceMint-based single-chain PoC protocol in the previous section,
most steps are natural implementations of our building block.
However, few issues should be carefully considered alongside.
To establish a commonly verifiable computation on the market weight $\mB$ in a decentralized environment,
we assume that the market weight of each chain is proportional to the summed weight of latest $\Delta$ consequent blocks by a constant factor $\mu$.
To make sure that each space proof is used for only one identity, there should be a pool $\Gamma$ of
space commits locally store in each node.
We have put the specific protocol in Fig.~\ref{fig:mpoc_protocol}
since details are not crucial to the roadmap of our paper.

\section{Framework Analysis}  \label{sec:6}
The realization of our functionality in Fig.~\ref{fig:mpoc_functionality} should satisfy
both safety and liveness. The analysis of basic safety and liveness is specific to the way of realization.
In our protocol to realize the functionality with $\s{lsize}=1$,
the basic safety and liveness properties of the realization are inherited from SpaceMint.
In this section, we focus on two higher-leveled properties of our general framework for a PoC multi-chain future, namely, incentive compatibility and the system security that a considerable fraction of global storage resource has to be held by an adversary to devastate any chain under the framework, even for newly launched chains or chains with the least market weight.

\subsection{Incentive Compatibility}
In this part, we prove that our framework is incentive-compatible.
Namely, participants achieve their greatest utility (the most expected revenue) with our desired strategy of storage resource partition.
Without loss of generality, we consider $\delta=0$ to simplify our proofs.
Also, we assume $\mbeta'=\mbeta$ from time to time to simplify proofs.
To begin with, we formally describe the partition strategy.
\begin{definition}[Capacity Resource Partition]
For each participant with total capacity $c$, all its admissible resource partitions form the space $\mathscr{D}_c^{M+1}\subset (\mathbb{R}^+)^{M+1}$ that $\sum_{i=0}^M {s_i} = c$ for each $\bm{s}=(s_0,s_1,\ldots,s_M)\in\mathscr{D}_c^{M+1}$.
Furthermore, we introduce $\mathscr{D}^{M+1}:=\cup_{c\in{\mathbb{R}^+}}\mathscr{D}_c^{M+1}$.
\end{definition}

Likewise, we denote $\so^M$ as the set of vectors $(\beta_1,\beta_2,\ldots,\beta_M)\subset (\mathbb{R}^+)^{M}$ of length $M$ that $\sum_{i=0}^M {\beta_i} = 1$.

%
%\begin{definition}[Admissible Resource Evaluation]
%An resource evaluation function $\val: \mathscr{D}^{M+1}\times \so^M \times [M] \rightarrow \mathbb{R}^+$ maps a resource partition $\bm{s}\in\mathscr{D}^{M+1}$ and a vector $\bm{\beta}=(\beta_1,\beta_2,\ldots,\beta_M)\in \so^M$ to a positive value for each chain $i\in[M]$.
%We call $\val(\bm{s},\bm{\beta},k)$ an admissible resource evaluation function if the following properties are held.
%\begin{itemize}
%  \item a
%  \item b
%  \item c
%\end{itemize}
%\end{definition}
The utility function is a mapping from a partition strategy to the expected revenue for each round.
To define the global utility function, we firstly introduce the chain-specific utility function.
\begin{definition}[Chain-Specific Utility Function]
The chain-specific utility function for chain $k$ is
$$
\val(\bm{s},\bm{B},k) :=
r\times w(\bm{s},\bm{B},k)\times \omega\mB_k.
$$
Intuitively, $r$ is a positive constant, $w(\bm{s},\bm{B},k)$ (proportionally) describes the probability of becoming the leader of each round, and the per-coin value of chain $k$ is proportional to the market weight of the chain (hence is $\omega\mB_k$ for a positive constant $\omega$).
Note that we have assume that the minted coins and transaction fees are the same for each chain and for each round, hence it is not a necessity to add in another factor for it in the multiplication.
\end{definition}
At a first glance, the formula above seems fit only into the scenario of $\s{lsize}=1$.
Actually, in the multi-leader case, each participant has $\s{lsize}$ times the chance to become an leader while each leader has only $\frac{1}{\s{lsize}}$ fraction of total revenue. By $\frac{\s{lsize}\cdot\bm{s}_k}{\bm{B}'[k]} \cdot \frac{r}{\s{lsize}} \cdot \omega\mB'_k = \frac{\bm{s}_k}{\bm{B}'_k} \cdot r \cdot \omega\mB'_k$, the expected revenue remains identical to the one-leader case.

\begin{definition}[The Utility Function]
Thereby the final utility function is the sum of all chain-specific utility functions
$$
U(\mS,\mB) := \sum_{k=1}^M \val(\mS,\mB,k).
$$
\end{definition}

The optimal resource partition strategy is clearly depicted as below.
\begin{definition}[Optimal Resource Partition]
To a participant with total storage resource $c$, its optimal resource partition strategy in the environment where each chain $i$ has market weight $\mB_i$ is
$$
\widetilde{\opt}(c,\bm{B}) := \textnormal{argmax}_{\bm{s}\in\mathscr{D}_c^{M+1}} U(\mS,\mB)
$$
\end{definition}

We aim to show that our desired resource partition strategy optimizes the utility function.
That is to have the optimal resource partition equals our desired one.
\begin{theorem} \label{lemma:final}
For any $c\in\mathbb{R}^+$ and  any $\bm{B}\in \mathscr{D}^M_c$,
$$
 \widetilde{\opt}(c,\bm{B}) = \frac{c}{2} \circ \frac{c}{2} \bm{\beta}
$$
where $\bm{\beta} = \frac{\bm{B}}{||\bm{B}||}$.
\end{theorem}

To prove this theorem, we at first introduce two lemmas.

\begin{lemma} \label{lemma:B}
For any positive integer $n$, positive values $k,\ (\mbeta_1,\mbeta_2,\ldots,\mbeta_n)$ with $\sum_{i=1}^n\mbeta_i=1$,
function $F(\bm{x}):= - \sum_{i=1}^n \mbeta_i \cdot \frac{x_i}{k\mbeta_i+x_i}$ ($\bm{x}=(x_0,x_1,\ldots,x_n)$)
achieves its minimum in $\hat{\bmx}=(\frac{c}{2}, \frac{\mbeta_1 c}{2},\ldots, \frac{\mbeta_n c}{2})$
subject to
\begin{itemize}
\item $g_i(\bm{x}) := x_i-\mbeta_i x_0 \le 0$ for each $i\in[n]$,
\item $g_{n+i+1}(\bmx):=-x_i\le 0$ for each $i\in[n]\cup\{0\}$,
\item $h(\bm{x}):=\sum_{i=0}^n x_i-c=0$.
\end{itemize}
\end{lemma}
\begin{proof}
To begin with, we show that $F$ is convex  on $(\mathbb{R}^+)^{n+1}$.
It is easy to observe that the Hessian matrix of $F$ is diagonal since
$\frac{\partial^2 F}{\partial x_i\partial x_j} =0$ for each $i\ne j$.
Each element in the diagonal
$$
\frac{\partial^2 F}{\partial x_i^2}=\frac{2\mbeta_i}{(k\mbeta_i+x_i)^2}\left( 1- \frac{x_i}{k\mbeta_i+x_i}\right)
$$
is positive so the Hessian matrix is semi-definite and $F$ is convex on $(\mathbb{R}^+)^{n+1}$.
Since $F$ and each $g_i$ are convex and $h$ is an affine function,
this turns out to be a convex optimization and
any minimal value of $F$ is its minimum.
According to Karush-Kuhn-Tucker~(KKT) conditions, we only need to show that
\begin{itemize}
\item $\triangledown L(\hat{\bmx})=\bm{0}$,
\item $\mu_i g_i(\hat{\bmx})=0$ for all $i\in[2n+1]$,
\end{itemize}
where
\begin{align}
L(\bmx) & =F(\bmx) + \sum_{i=1}^{2n+1}{\mu_i g_i(\bmx)} + \nu h(\bmx) \\
&= - \sum_{i=1}^n \mbeta_i\cdot \frac{x_i}{k\mbeta_i+x_i} + \sum_{i=1}^n \mu_i \left(x_i-\mbeta_ix_0\right) + \nu \left(\sum_{i=0}^n x_i -c\right),
\end{align}
$\mu_i=\nu=\frac{k}{c(k+c/2)}$ for each $i\in[n]$.
For each natural number $i\le n$, $\mu_{n+i+1}$ are set to be zero so $\sum_{i=n+1}^{2n+1}\mu_ig_i(\bmx)=0$.
In fact,
\begin{enumerate}
  \item From simple derivations, $$
        \frac{\partial L(\bmx)}{\partial x_0}=\sum_{i=1}^n -\mu_i\mbeta_i + \nu = \frac{k}{c(k+\frac{c}{2})}(1-\sum_{i=1}^n \mbeta_i)=0
        $$ and for each $i\in[n]$,
        $$
        \frac{\partial L(\bmx)}{\partial x_i}= -\mbeta_i \cdot \frac{k\mbeta_i}{x_i(k\mbeta_i+x_i)} +\mu_i +\nu.
        $$
        Thereby,
        \begin{align*}
        \frac{\partial L(\hat{\bmx})}{\partial x_i} &= -\mbeta_i \cdot \frac{k\mbeta_i}{\frac{\mbeta_i c}{2}(k\mbeta_i+\frac{\mbeta_i c}{2})} + \frac{2k}{c(k+\frac{c}{2})} \\
        &=  - \frac{2k}{c(k+\frac{ c}{2})} + \frac{2k}{c(k+\frac{c}{2})}=0.
        \end{align*}
        As a result, we conclude that $\triangledown L(\hat{\bmx})=\bm{0}$.
  \item The second condition is easy to hold since $g_i(\hat{\bmx})$ = 0 for each $i\le n$ and $\mu_i=0$ for each $i> n$.
\end{enumerate}
Therefore, $f$ achieves its minimum in $\hat{\bmx}$.
\end{proof}

Intuitively, Lemma.~\ref{lemma:B} alone has proved Theorem.~\ref{lemma:final} since by having $k=||\mB||$,
$$
U(\mS,\mB) = -F(\mS) \times r\omega||\mB||
$$
within the boundary in the lemma.
Since the Hessian matrix is semi-definite either inside the boundary~(shown in the lemma proof)
or outside the boundary~(easy to verify) and $U$  is obviously continuous everywhere and smooth a.e. except for a subset with zero Lebesgue measure, $U$ is almost convex and the local maximization is actually the global optimization.
However, to more strictly prove the theorem, we still require the following lemma.

\begin{lemma} \label{lemma:A}
For any space division strategy on $n$ chains, suppose $\mS^{A} = (\mS_{0}^{A}, \mS_{1}^{A}, \ldots, \mS_{n}^{A})$,
if $\exists i \in [n]: \mS_{i}^{A} > \mbeta_{i}\mS_{0}^{A}$ , then there always exists another strategy whose utility is better than $\mS^{A}$.
\end{lemma}
\begin{proof}
If there exists $i \in [n]: \mS_{i}^{A} > \mbeta_{i}\mS_{0}^{A}$, then there must exist a set of index $I^{A} = \{ a_{1}, a_{2}, \ldots, a_{exc} \} \subset [n]$ where $\forall i \in I^{A}$, $\mS_{i}^{A} > \mbeta_{i}\mS_{0}^{A}$ and $\forall i \notin I^{A}$, $\mS_{i}^{A} \le \mbeta_{i}\mS_{0}^{A}$ both satisfy.
Without loss of generality, we assume
$$\frac{\mS_{a_{1}}^{A}}{\mbeta_{a_{1}}} \ge \frac{\mS_{a_{2}}^{A}}{\mbeta_{a_{2}}} \ge \ldots \ge \frac{\mS_{a_{exc}}^{A}}{\mbeta_{a_{exc}}} > \mS_{0}^{A}.$$

Consider another strategy $\mS^{A+} = (\mS_{0}^{A+}, \mS_{1}^{A+}, \dots, \mS_{n}^{A+})$. When comparing $\mS^{A+}$ with $\mS^{A}$, the new strategy just subtracts some space division from $\mS_{a_{1}}$ and adds it to $\mS_{0}$. In detail, $\mS^{A+}$ can be formally expressed as
$$ \mS_{i}^{A+}=\left\{
\begin{array}{rcl}
&\frac{\mbeta_{a_1}}{\mbeta_{a_1}+1}(\mS_{0}^{A}+\mS_{a_{1}}^{A})  &\qquad i = a_{1}\\
&\frac{1}{\mbeta_{a_1}+1}(\mS_{0}^{A}+\mS_{a_{1}}^{A})   &\qquad i = 0\\
&\mS_{i}^{A}   &\qquad o.w.\\
\end{array} \right. $$

Notice that if $\mS^{A}$ is a valid space division strategy, then $\mS^{A+}$ is also a valid strategy, as $\sum_{i = 0}^{n}\mS_{i}^{A+} = \sum_{i = 0}^{n}\mS_{i}^{A} = c$.

Now we compare the utility of $\mS_{i}^{A}$ with that of $\mS_{i}^{A+}$. Consider the utility function $U$ for a space division strategy. If we denote the truly effective space division for chain $i$ as $\s{seff}_{i}$, where $$\s{seff}_{i} = \min \{ \mbeta_{i}\mS_{0},\mS_{i} \}.$$
Then, it is obvious that the utility function is monotonic for every $\s{seff}_{i}$ where $i \in [n]$. As a result, if $\forall i \in [n], \s{seff}_{i}^{A+} \ge \s{seff}_{i}^{A}$ and $\exists i \in [n], \s{seff}_{i}^{A+} > \s{seff}_{i}^{A}$, then the utility of strategy $\mS^{A+}$ is better than that of $\mS^{A}$.

In fact, when $i = a_{1}$, we have $\s{seff}_{i}^{A} = \mbeta_{a_{1}}\mS_{0}^{A}$, while $\s{seff}_{i}^{A+} = \mS_{i}^{A+} = \frac{\mbeta_{a_1}}{\beta_{a_1}+1}(\mS_{0}^{A}+\mS_{a_{1}}^{A})$. Since $\frac{\mS_{a_{1}}^{A}}{\mbeta_{a_{1}}} > \mS_{0}^{A}$, it comes out that $\s{seff}_{a_{1}}^{A+} > \s{seff}_{a_{1}}^{A}$.
When $i \in I^{A}$ and $i \ne a_{1}$, given that we have assumed $\mS_{i}^{A} > \mbeta_{i}\mS_{0}^{A}$, then $\s{seff}_{i}^{A} = \mbeta_{i}\mS_{0}^{A}$. Since
$$\mS_{0}^{A+} = \frac{1}{\mbeta_{a_1}+1}(\mS_{0}^{A}+\mS_{a_{1}}^{A}) > \frac{1}{\mbeta_{a_1}+1}(\mS_{0}^{A}+\mbeta_{a_{1}}\mS_{0}^{A}) = \mS_{0}^{A}, $$
 we have
 $$\s{seff}_{i}^{A+} = \min \{ \mbeta_{i}\mS_{0}^{A+},\mS_{i}^{A+} \} > \mbeta_{i}\mS_{0}^{A} = \s{seff}_{i}^{A}.$$
 When $i \in [n]$ and $i \notin I^{A}$, it is obvious that $\s{seff}_{i}^{A+} = \mS_{i}^{A+} = \mS_{i}^{A} = \s{seff}_{i}^{A}$. And it concludes that the utility of $\mS^{A+}$ is better than that of $\mS^{A}$.

\end{proof}

Now we are allowed to prove Theorem.~\ref{lemma:final}.
\begin{proof}
  We partite $\mathscr{D}_c^{M+1}=D\uplus D'$ into two domains where
  $$
  D= \left\{ \mS\in \mathscr{D}_c^{M+1} | \forall i\in[M].\ \mS_i \le \mbeta_i \mS_0 \right\}.
  $$
   Lemma.~\ref{lemma:A} tells that each strategy in $D'$ can be emulated by a strategy in $D$.
   Therefore, $D$ includes the optimal partition strategy within $\mathscr{D}_c^{M+1}$.
  By having $k=||\mB||$, $U(\mS,\mB) = -F(\mS) \times r\omega||\mB||$ in Lemma.~\ref{lemma:B} within $D$
  and so forth the optimal strategy among $D$ is our desired partition.
\end{proof}

Based on Theorem.~\ref{lemma:final}, we find that the optimization of strategy is indeed independent from the total resource $c$. Hence we conclude that for any participant, the optimal strategy is to divide half resource for the shared proof and divide the rest part according to the market weight of each chain
$$
\opt(\bm{\beta}) := \frac{1}{c} \widetilde{\opt}(c,\bm{B}) = \frac{1}{2} \circ \frac{1}{2} \bm{\beta}.
$$

\subsection{System Security}
In this part, we show the difficulty of devastating a chain (say, $i^\textnormal{th}$ chain of market weight $\bm{B}_i=b$) with $\frac{b}{||\bm{B}||}=\beta$ fraction of total
market weight.
To devastate chain $i$, the adversary should occupy the total market weight with total storage resource over $\alpha b$\footnote{Most existing systems ask for $\alpha>1/3$, but we treat $\alpha$ as a tunable parameter to allow flexibility.}.
That is to have $\bm{s}_i = \alpha (b+\bm{s}_i)$ and at the same time
$(1+\delta)\bm{\beta}'_i \bm{s}_0 \ge \bm{s}_i$
where
$\bm{\beta}'_i = \frac{b+\bm{s}_i}{||\bm{B}||+\bm{s}_i}$.
Thereby,
$$
\bm{s}_i = \frac{\alpha}{1-\alpha}b,
$$
and hence
\begin{align*}
\bm{s}_0 &\ge \frac{\bm{s}_i}{(1+\delta)\bm{\beta}'_i} =
\frac{\bm{s}_i (||\bm{B}||+\bm{s}_i)}{(1+\delta)(b+\bm{s}_i)}
\\
&\ge \frac{\frac{\alpha}{1-\alpha}b ||\bm{B}||}{(1+\delta)(b+\frac{1-\alpha}{\alpha}b)} = \frac{\alpha^2 }{(1+\delta)(1-\alpha)}||\bm{B}||.
\end{align*}
%%%%%%%%%%%%%
For $\alpha>\frac{1}{3}$ and $\delta=\frac{1}{10}$,
$\bm{s}_0\ge \frac{5}{33}||\bm{B}|| > 15\% ||\bm{B}||$.
This tells that regardless of the market weight of the chain, the adversary has to devote  more than $15\%$ global storage resource ($45\%$ to $\alpha>1/2$) to devastate a chain with even the slightest market weight.
%\section{The Detailed Protocol}

\section{Conclusion}
In this paper, we have proposed a novel multi-chain scheme from the inherited merit of proof-of-space.
With our framework, the same storage source contributes simultaneously to multiple blockchains and
 newly set up blockchains are hard to be devastated.
In the future, we look forward to the flourishing development of PoC-based blockchains and having our framework implemented on them.
Also, although we have shown to realize this framework via pebbling graph-styled proof-of-space schemes,
we also expect its application on PoC schemes of other styles.

\bibliographystyle{unsrt}
\bibliography{dblp0}

\begin{thebibliography}{10}

\bibitem{bitcoin}
Satoshi Nakamoto.
\newblock Bitcoin: A peer-to-peer electronic cash system, 2008.

\bibitem{DBLP:conf/crypto/DworkN92}
Cynthia Dwork and Moni Naor.
\newblock Pricing via processing or combatting junk mail.
\newblock In {\em Advances in Cryptology - {CRYPTO} '92, 12th Annual
  International Cryptology Conference, Santa Barbara, California, USA, August
  16-20, 1992, Proceedings}, pages 139--147, 1992.

\bibitem{pos}
{QuantumMechanic~et~al.}
\newblock Proof of stake instead of proof of work.
\newblock Bitcoin forum, 2011.
\newblock \url{https://bitcointalk.org/index.php?topic=27787.0}.

\bibitem{DBLP:journals/sigmetrics/BentovLMR14}
Iddo Bentov, Charles Lee, Alex Mizrahi, and Meni Rosenfeld.
\newblock Proof of activity: Extending bitcoin's proof of work via proof of
  stake [extended abstract]y.
\newblock {\em {SIGMETRICS} Performance Evaluation Review}, 42(3):34--37, 2014.

\bibitem{DBLP:journals/iacr/BentovPS16a}
Iddo Bentov, Rafael Pass, and Elaine Shi.
\newblock Snow white: Provably secure proofs of stake.
\newblock {\em {IACR} Cryptology ePrint Archive}, 2016:919, 2016.

\bibitem{DBLP:conf/asiacrypt/PassS17}
Rafael Pass and Elaine Shi.
\newblock The sleepy model of consensus.
\newblock In {\em Advances in Cryptology - {ASIACRYPT} 2017 - 23rd
  International Conference on the Theory and Applications of Cryptology and
  Information Security, Hong Kong, China, December 3-7, 2017, Proceedings, Part
  {II}}, pages 380--409, 2017.

\bibitem{DBLP:conf/crypto/KiayiasRDO17}
Aggelos Kiayias, Alexander Russell, Bernardo David, and Roman Oliynykov.
\newblock Ouroboros: {A} provably secure proof-of-stake blockchain protocol.
\newblock In {\em Advances in Cryptology - {CRYPTO} 2017 - 37th Annual
  International Cryptology Conference, Santa Barbara, CA, USA, August 20-24,
  2017, Proceedings, Part {I}}, pages 357--388, 2017.

\bibitem{DBLP:conf/eurocrypt/DavidGKR18}
Bernardo David, Peter Gazi, Aggelos Kiayias, and Alexander Russell.
\newblock Ouroboros praos: An adaptively-secure, semi-synchronous
  proof-of-stake blockchain.
\newblock In {\em Advances in Cryptology - {EUROCRYPT} 2018 - 37th Annual
  International Conference on the Theory and Applications of Cryptographic
  Techniques, Tel Aviv, Israel, April 29 - May 3, 2018 Proceedings, Part {II}},
  pages 66--98, 2018.

\bibitem{DBLP:conf/ccs/BadertscherGKRZ18}
Christian Badertscher, Peter Gazi, Aggelos Kiayias, Alexander Russell, and
  Vassilis Zikas.
\newblock Ouroboros genesis: Composable proof-of-stake blockchains with dynamic
  availability.
\newblock In {\em Proceedings of the 2018 {ACM} {SIGSAC} Conference on Computer
  and Communications Security, {CCS} 2018, Toronto, ON, Canada, October 15-19,
  2018}, pages 913--930, 2018.

\bibitem{DBLP:conf/crypto/DworkNW05}
Cynthia Dwork, Moni Naor, and Hoeteck Wee.
\newblock Pebbling and proofs of work.
\newblock In {\em Advances in Cryptology - {CRYPTO} 2005: 25th Annual
  International Cryptology Conference, Santa Barbara, California, USA, August
  14-18, 2005, Proceedings}, pages 37--54, 2005.

\bibitem{DBLP:conf/crypto/DziembowskiKW11}
Stefan Dziembowski, Tomasz Kazana, and Daniel Wichs.
\newblock Key-evolution schemes resilient to space-bounded leakage.
\newblock In {\em Advances in Cryptology - {CRYPTO} 2011 - 31st Annual
  Cryptology Conference, Santa Barbara, CA, USA, August 14-18, 2011.
  Proceedings}, pages 335--353, 2011.

\bibitem{DBLP:conf/tcc/DziembowskiKW11}
Stefan Dziembowski, Tomasz Kazana, and Daniel Wichs.
\newblock One-time computable self-erasing functions.
\newblock In {\em Theory of Cryptography - 8th Theory of Cryptography
  Conference, {TCC} 2011, Providence, RI, USA, March 28-30, 2011. Proceedings},
  pages 125--143, 2011.

\bibitem{spacemint}
Georg Fuchsbauer.
\newblock Spacemint: {A} cryptocurrency based on proofs of space.
\newblock {\em {ERCIM} News}, 2017(110), 2017.

\bibitem{8}
Giuseppe Ateniese, Ilario Bonacina, Antonio Faonio, and Nicola Galesi.
\newblock Proofs of space: When space is of the essence.
\newblock In {\em Security and Cryptography for Networks - 9th International
  Conference, {SCN} 2014, Amalfi, Italy, September 3-5, 2014. Proceedings},
  pages 538--557, 2014.

\bibitem{2}
Stefan Dziembowski, Sebastian Faust, Vladimir Kolmogorov, and Krzysztof
  Pietrzak.
\newblock Proofs of space.
\newblock In {\em Advances in Cryptology - {CRYPTO} 2015 - 35th Annual
  Cryptology Conference, Santa Barbara, CA, USA, August 16-20, 2015,
  Proceedings, Part {II}}, pages 585--605, 2015.

\bibitem{7}
COLIN PERCIVAL.
\newblock Stronger key derivation via sequential memory-hard functions.
\newblock 01 2009.

\bibitem{9}
Ari Juels and Burton S.~Kaliski Jr.
\newblock Pors: proofs of retrievability for large files.
\newblock In {\em Proceedings of the 2007 {ACM} Conference on Computer and
  Communications Security, {CCS} 2007, Alexandria, Virginia, USA, October
  28-31, 2007}, pages 584--597, 2007.

\bibitem{3}
Wolfgang~J. Paul, Robert~Endre Tarjan, and James~R. Celoni.
\newblock Space bounds for a game on graphs.
\newblock {\em Mathematical Systems Theory}, 10:239--251, 1977.

\bibitem{DBLP:journals/im/AlonC03}
Noga Alon and Michael~R. Capalbo.
\newblock Smaller explicit superconcentrators.
\newblock {\em Internet Mathematics}, 1(2):151--163, 2003.

\bibitem{DBLP:conf/sirocco/Schoning97}
Uwe Sch{\"{o}}ning.
\newblock Better expanders and superconcentrators by kolmogorov complexity.
\newblock In {\em SIROCCO'97, 4th International Colloquium on Structural
  Information {\&} Communication Complexity, Monte Verita, Ascona, Switzerland,
  July 24-26, 1997}, pages 138--150, 1997.

\bibitem{DBLP:journals/arscom/KolmogorovR18}
Vladimir Kolmogorov and Michal Rolinek.
\newblock Superconcentrators of density 25.3.
\newblock {\em Ars Comb.}, 141:269--304, 2018.

\bibitem{DBLP:journals/ppl/Haglin95}
David~J. Haglin.
\newblock Bipartite expander matching is in {NC}.
\newblock {\em Parallel Processing Letters}, 5:413--420, 1995.

\bibitem{DBLP:journals/dm/Thomason89a}
Andrew Thomason.
\newblock Dense expanders and pseudo-random bipartite graphs.
\newblock {\em Discrete Mathematics}, 75(1-3):381--386, 1989.

\bibitem{4}
Paul Erdoes, Ronald~L. Graham, , and Endre Szemeredi.
\newblock On sparse graphs with dense long paths.
\newblock Technical report, Stanford, CA, USA,, 1975.

\bibitem{DBLP:conf/eurocrypt/AlwenBP17}
Jo{\"{e}}l Alwen, Jeremiah Blocki, and Krzysztof Pietrzak.
\newblock Depth-robust graphs and their cumulative memory complexity.
\newblock In {\em Advances in Cryptology - {EUROCRYPT} 2017 - 36th Annual
  International Conference on the Theory and Applications of Cryptographic
  Techniques, Paris, France, April 30 - May 4, 2017, Proceedings, Part {III}},
  pages 3--32, 2017.

\bibitem{5}
Ling Ren and Srinivas Devadas.
\newblock Proof of space from stacked expanders.
\newblock In {\em Theory of Cryptography - 14th International Conference, {TCC}
  2016-B, Beijing, China, October 31 - November 3, 2016, Proceedings, Part
  {I}}, pages 262--285, 2016.

\bibitem{6}
Krzysztof Pietrzak.
\newblock Proofs of catalytic space.
\newblock In {\em 10th Innovations in Theoretical Computer Science Conference,
  {ITCS} 2019, January 10-12, 2019, San Diego, California, {USA}}, pages
  59:1--59:25, 2019.

\bibitem{DBLP:conf/ccs/BellareR93}
Mihir Bellare and Phillip Rogaway.
\newblock Random oracles are practical: {A} paradigm for designing efficient
  protocols.
\newblock In {\em {CCS} '93, Proceedings of the 1st {ACM} Conference on
  Computer and Communications Security, Fairfax, Virginia, USA, November 3-5,
  1993.}, pages 62--73, 1993.

\bibitem{10}
Yevgeniy Dodis, Siyao Guo, and Jonathan Katz.
\newblock Fixing cracks in the concrete: Random oracles with auxiliary input,
  revisited.
\newblock In {\em Advances in Cryptology - {EUROCRYPT} 2017 - 36th Annual
  International Conference on the Theory and Applications of Cryptographic
  Techniques, Paris, France, April 30 - May 4, 2017, Proceedings, Part {II}},
  pages 473--495, 2017.

\bibitem{posw}
Bram Cohen and Krzysztof Pietrzak.
\newblock Simple proofs of sequential work.
\newblock In {\em Advances in Cryptology - {EUROCRYPT} 2018 - 37th Annual
  International Conference on the Theory and Applications of Cryptographic
  Techniques, Tel Aviv, Israel, April 29 - May 3, 2018 Proceedings, Part {II}},
  pages 451--467, 2018.

\bibitem{DBLP:books/daglib/0016239}
John~E. Savage.
\newblock {\em Models of computation - exploring the power of computing}.
\newblock Addison-Wesley, 1998.

\end{thebibliography}

\end{document}